\documentclass[11pt,twoside,english]{article}
\usepackage[T1]{fontenc}
\usepackage[latin9]{inputenc}
\usepackage{url}
\usepackage{amsmath}
\usepackage{amssymb}
\usepackage{amsthm}\usepackage{amstext}
\usepackage{jsat}

\theoremstyle{definition}
\newtheorem{definition}{Definition}[section]
\newtheorem{example}{Example}

\newtheorem{proposition}{Proposition}[section]
\theoremstyle{remark}

\newtheorem*{remark*}{Remark}
\newtheorem*{notation}{Notation}

\numberwithin{equation}{section}

\makeatother

\usepackage{babel}

\begin{document}

%%%     Titre/Auteur	%%%

\title{About Inverse 3--SAT}

\author{\name Xavier Labouze
	\email xavier.labouze@u-psud.fr\\
	\addr D\'epartement de Math\'ematiques,\\
	Facult\'e de Pharmacie, Universit\'e Paris 11,\\
	 F--92296 Ch\^atenay-Malabry Cedex, France}	

%\thanks{}
%\date{15 March 2013}
%\dedicatory{\`A Karine}

\maketitle
%%%    Abstract 	%%%

\begin{abstract}

	The Inverse 3--SAT problem is known to be coNP Complete: Given $\phi$
	a set of models on $n$ variables, is there a 3--CNF formula such
	that $\phi$ is its exact set of models ? An immediate candidate formula
	$F_{\phi}^{3}$ arises, which is the conjunction of all 3--clauses
	satisfied by all models in $\phi$. The (co)Inverse 3--SAT problem
	can then be resumed: Given $\phi$ a set of models on $n$ variables,
	is there a model of $F_{\phi}^{3}\notin\phi$ ? \\
	This article uses two important intermediate results: 1- The candidate
	formula can be easily (i.e. in polynomial time) transformed into an
	equivalent formula $F_{\phi}$ which is 3--closed under resolution.
	A crucial property of $F_{\phi}$ is that the induced formula $F_{\phi|I}$
	by applying any partial assignment $I$ of the $n$ variables to $F_{\phi}$
	is unsatisfiable iff its 3--closure contains the empty clause. 2-
	A set of partial assignments (of polynomial size) which subsume all
	assignments $\notin\phi$ can be easily computed.\\
	 The (co)Inverse 3--SAT question is then equivalent to decide
	whether it exists a partial assignment $I\notin\phi$ such that the
	3--closure of $F_{\phi|I}$ does not contain the empty clause.
 
\end{abstract}


%\published{September 2013}{October 2004}{November 2004} % Text of article.

\section{Introduction}

	The satisfiability problem has been one of the most studied problems
	in computational complexity \cite{Cook1971,Garey1979,Gu1997,Karp1972,Levin1973,Papadimitriou1994,Valiant1979}.
	Kavvadias and Sideri have shown that the Inverse 3--SAT problem is
	coNP Complete \cite{kavvadias1996}: Given $\phi$ a set of models
	on $n$ variables, is there a 3--CNF formula such that $\phi$ is
	its exact set of models ? An immediate candidate 3--CNF formula $F_{\phi}^{3}$
	arises which is the set of all 3--clauses satisfied by all models
	in $\phi$. Since $F_{\phi}^{3}$ is the most restricted 3--CNF formula
	(in term of its model set) which is satisfied by all models in $\phi$,
	the (co)Inverse 3--SAT problem can then be defined: Given $\phi$
	a set of models on $n$ variables, is there a model of $F_{\phi}^{3}\notin\phi$
	? The properties of $F_{\phi}^{3}$ will bring a new interesting way
	to solve the Inverse 3--SAT problem.
	
	In the next part of the article, all needed notations will be defined.
	In section 3, the main ideas of the algorithm presented in section 4 will be developped.

\section{Preliminaries}

\paragraph{3--CNF formula.}

	A CNF propositionnal formula $F$ is regarded in the standard way
	as a set of clauses, where each clause is regarded as a set of literals,
	and each literal as a boolean variable or its negation. Whether $x$
	is a positive or a negative literal, $\bar{x}$ denotes its complement.
	The size of a set $A$ (denoted $|A|$) is the number of its elements.
	A 3--clause is a clause of size 3. The 3--clause $c=\{x,y,z\}$ is
	denoted $(xyz)$. $c\smallsetminus\{x\}$ is the clause $(yz)$. The
	empty clause, denoted $(\varnothing)$, is equivalent to $false$.
	A 3--CNF formula is a CNF formula containing at least one 3--clause.

\paragraph{Assignment.}

	Let $F$ be a 3--CNF formula on $n$ variables $\{x_{1},x_{2}\dots x_{n}\}$.
	Each variable $x_{i}$ can be assigned to the value $v_{i}$. A (total)
	{\em assignment} of the $n$ variables is a set of $n$ values
	$\{v_{1},v_{2}\dots v_{n}\}$, where the value $v_{i}$ is assigned
	to the variable $x_{i}$. A value $v$ is equal to 0 ($false$) or
	1 ($true$), the opposite of the value $v$, $\bar{v}=1-v$. A clause
	of $F$ is satisfied when at least one of its literals is set (assigned)
	to {\em true}. $F$ is satisfiable if it exists a truth assignment
	of the $n$ variables which satisfies all its clauses. Such a truth
	assignment is called a {\em model}. A {\em partial} assignment
	on $k$ variables is the subset of a total assignment restricted to
	the values of the choosen $k$ variables ($k\le n$). 

	\begin{definition}

		Given $F$ a 3--CNF formula on $n$ variables $\{x_{1},x_{2}\dots x_{n}\}$;
		$c$, a clause in $F$; $I$, a partial assignment of $k$ variables
		among $(x_{i})\ (k\le n)$. 

		\begin{enumerate}

			\item 
			Let $F_{|I}$ be the induced formula by applying $I$ to $F$: Any
			clause that contains a literal which evaluates to $true$ under $I$
			is deleted from the formula and any literals that evaluate to $false$
			under $I$ are deleted from all clauses - the clauses that become
			empty by this deletion remain in the formula as the empty clause. 

			\item 
			Let $c_{|I}$ be the induced clause by applying $I$ to $c$: If $c$
			contains a literal which evaluates to $true$ under $I$ then $c_{|I}=true$;
			If $c$ contains a subset $A$ of literals all set to $false$ under
			$I$ then $c_{|I}=c\smallsetminus A$; If $c$ does not contain any
			literal set by $I$ then $c_{|I}=c$. 

		\end{enumerate}

	\end{definition}

\paragraph{Subsumption.}

	A clause $c$ is said to {\em subsume} a clause $d$, and $d$
	{\em is subsumed} by $c$, if the literals of $c$ are a subset
	of those of $d$ (each clause subsumes itself then). A (partial) assignment
	$I$ is said to {\em subsume} a (partial) assignment $J$, and
	$J$ {\em is subsumed} by $I$, if the values of $I$ are a subset
	of those of $J$.

\paragraph{Resolution.}

	Two clauses, $c_{1}=(Ax)$ and $c_{2}=(B\bar{x})$, can be resolved
	in a third clause $c=(AB)$, so called {\em resolvent} ($c_{1}$
	and $c_{2}$ are the {\em operands}), where $A$ and $B$ are two
	subsets of literals. A {\em 3--limited resolution} is a resolution
	in which the resolvent (so called 3--limited resolvent) and the operands
	have at most 3 literals.\\
	 A CNF formula $F$ is said to be {\em closed under resolution}
	{[}respectively {\em 3--limited closed under resolution}{]} (or
	just {\em closed} {[}resp. {\em 3--limited closed}{]}) if no
	clause of $F$ is subsumed by a different clause of $F$, and the
	resolvent {[}resp. 3--limited resolvent{]} of each pair of resolvable
	clauses is subsumed by some clause of $F$.\\
	 The {\em closure} {[}resp. {\em 3--limited closure}{]}
	of a CNF formula $F$ is the CNF formula (denoted $F^{c}$ {[}resp.
	3L--$F^{c}${]}) that derived from $F$ by a series of resolutions
	{[}resp. 3--limited resolutions{]} (which add clauses) and subsumptions
	(which delete clauses), and is closed {[}resp. 3--limited closed{]}.
	Both closure and 3--limited closure are unique
	\cite{Vangelder1995}. In the same paper \cite{Vangelder1995}, the
	3--limited closure of a CNF formula has been shown to be computable
	in polynomial time.\\
	 $F^{c}$ can be separate into 2 disjoint subsets: $F^{c}=$ 3--$F^{c}\cup F^{r}$,
	where 3--$F^{c}$ is the {\em 3-closure} of $F$, i. e. the subset
	of $F^{c}$ containing only clauses of size 3 or less ( each clause
	of 3L--$F^{c}$ is then subsumed by some clause of 3--$F^{c}$), and $F^{r}$
	contains clauses of size 4 or more.

\section{Discussion before the algorithm}

	Given $\phi=\{m_{1},m_{2}\dots m_{|\phi|}\}$, a set of $|\phi|$
	models on $n$ variables $(x_{i})$$_{i\le1\le n}$ (an element of
	$\phi$ will be called either assignment or model or simply element
	according to the context). Let $F_{\phi}^{3}$ the set of all 3--clauses
	satisfied by all models in $\phi$.
		
	\subsection
	{The 3--closure of $F_{\phi|I}$ can be computed in polynomial time}
	
		Given $I$, a partial assignment of $k$ variables among $(x_{i})\ (k\le n)$.
		
		\begin{proposition} 
			The 3--closure of $F_{\phi}^{3}$ can be computed
			in polynomial time. 
		\end{proposition}
		
		\begin{proof}
			Since $F_{\phi}^{3}$ contains all 3--clauses satisfied by all models
			in $\phi$, all possible 3--clauses implied by $F_{\phi}^{3}$ are
			in $F_{\phi}^{3}$. Since any resolvent of size 2 or less results
			from the resolution of clauses of size 3 or less, the 3--closure under
			resolution of $F_{\phi}^{3}$ can be computed in polynomial time.
		\end{proof}
		
		\begin{notation} 
			Call $F_{\phi}$ (or $F$ if it is not confusing)
			the 3--closure of $F_{\phi}^{3}$. 
		\end{notation}
		
		\begin{remark*} 
			(1) Each clause of $F_{\phi}^{3}$ is subsumed by
			a clause in $F_{\phi}$ and $F_{\phi}$ is equivalent to $F_{\phi}^{3}$
			. (2) As $F_{\phi}^{c}=F_{\phi}\cup F_{\phi}^{r}$ then all clauses
			of $F_{\phi}^{r}$ result from resolution of clauses of $F_{\phi}^{3}$
			or some iterated resolvents of clauses of $F_{\phi}^{3}$. 
		\end{remark*}

\begin{example}
		
			\label{ex1} Take $n=5$ and 8 models $(m_{i})_{1\leq i\leq8}$ in
			$\phi$.\\
			 $\phi=\{00111,01011,10101,11100,11111,10011,01101,00100\}$\\
			 By gathering all 3--clauses satisfied by all models of $\phi$:
			\begin{align*}
			F_{\phi}^{3}= & (x_{1}x_{2}x_{3})(\bar{x}_{1}\bar{x}_{2}x_{3})(x_{1}\bar{x}_{2}x_{5})(\bar{x}_{1}x_{2}x_{5})(x_{1}x_{3}x_{4})(\bar{x}_{1}x_{3}x_{4})(x_{1}x_{3}x_{5})(\bar{x}_{1}x_{3}x_{5})(x_{1}\bar{x}_{4}x_{5})\\
			 & (\bar{x}_{1}\bar{x}_{4}x_{5})(x_{2}x_{3}x_{4})(\bar{x}_{2}x_{3}x_{4})(x_{2}x_{3}x_{5})(\bar{x}_{2}x_{3}x_{5})(x_{2}\bar{x}_{4}x_{5})(\bar{x}_{2}\bar{x}_{4}x_{5})(x_{3}x_{4}x_{5})(x_{3}x_{4}\bar{x}_{5})\\
			 & (x_{3}\bar{x}_{4}x_{5})(\bar{x}_{3}\bar{x}_{4}x_{5}) \end{align*}
			
			Its 3-closure is:\

			 $F_{\phi}=(x_{1}x_{2}x_{3})(\bar{x}_{1}\bar{x}_{2}x_{3})(x_{1}\bar{x}_{2}x_{5})(\bar{x}_{1}x_{2}x_{5})(x_{3}x_{4})(x_{3}x_{5})(\bar{x}_{4}x_{5})$
		
		\end{example}
		
		\begin{proposition} 
			Given $I$, a partial assignment of $k$ variables
			among $(x_{i})\ (k\le n)$, the 3--closure of $F_{\phi|I}$ is computable
			in polynomial time. 
		\end{proposition}
		
		\begin{proof} By recurrence.
		
			Let $R_{|I}$ the 3--limited closure of $F_{\phi}\cup F_{\phi|I}$,
			i.e. the set of clauses easily reachable from $F_{\phi}$ or $F_{\phi|I}$.
			Given $c$ a clause implied by $F_{\phi}$, it exists at least one
			subset of $R_{|I}$ whose clauses imply $c$. Name $R_{c}$ such a
			subset. \\

			Let $P(k)$ the following property :\\
			 $P(k):$ For all $c$ implied by $F_{\phi}$ such that $|c_{|I}|\le3,$\\
			 $[\exists R_{c}\subseteq R_{|I}$ such that $|R_{c}|\le k\Rightarrow c_{|I}$
			is subsumed by some clause $\in$ 3L--$F_{\phi|I}^{c}]$\\
			 (3L--$F_{\phi|I}^{c}$ is the 3--limited closure of $F_{\phi|I})$
			\\

			Here does the recurrence begin. \\
			 Given $c$ implied by $F_{\phi}$ such that $|c_{|I}|\le3$,
			i.e. $c_{|I}\in$ the 3--closure of $F_{\phi|I}$.

			\begin{enumerate}
			
				\item 
				$k=1$. If $\exists R_{c}\subseteq R_{|I}/|R_{c}|=1$ then $R_{c}=\{d\}$
				($d\in F_{\phi}\cup F_{\phi|I}$ subsumes $c$) and $c_{|I}$ is subsumed
				by $d_{|I}\in F_{\phi|I}$ (note that any clause of $F_{\phi|I}$
				is subsumed by some clause of $G_{\phi|I}$). Thus $P(1)$.
				
				\item 
				Suppose $P(k)$ for $k\ge1$. If $\exists R_{c}\subseteq R_{|I}$
				such that $|R_{c}|\le k+1$ (and no other $R_{c}$ of size 1 such
				that $c\notin F_{\phi}$ and $|c|>3$) then suppose $c=(\alpha \beta \gamma L_{I})$
				where $\alpha ,\beta ,\gamma$ are literals not set by $I$ and $L_{I}$ is a subset
				of literals all evaluate to 0 under $I (L_I \neq \varnothing)$, i.e $c_{|I}=(\alpha \beta \gamma)$, with
				$\alpha ,\beta ,\gamma$ not necessarily different. 
				
				\item 
				Remove a clause $d_{i}$ from $R_{c}$ such that $|d_{i|I}|<|d_{i}|\le3$,
				in other words, such that $d_{i}$ contains some literal from $L_{I}$
				(there is at least one such clause in $R_{c}$ since $L_I \neq \varnothing$) and $|d_{i|I}|\leq 2$.
				
				\item 
				The size of the remaining set $R_{c}\setminus d_{i}$ is $\le k$.
				If a certain clause $c'=(\alpha \beta \gamma L'_{I})$ is implied by $R_{c}\setminus d_{i}$
				(where $L'_{I}$ is a subset of literals all evaluate to 0 under $I$)
				then $|c'_{|I}|=3$ and $\exists R_{c'}=R_{c}\setminus d_{i}\subseteq R_{|I}$
				such that $|R_{c'}|\le k$. By $P(k)$, $c'_{|I}=(\alpha \beta \gamma)$ is then subsumed
				by some clause $\in$ 3L--$F_{\phi|I}^{c}$, inducing $P(k+1)$ for
				$c$.

				\item
				If $d_{i|I}$ contains $\bar \alpha$ or $\bar \beta$ or $\bar \gamma$ then $d_{i|I}$ is useless to imply [some clause subsuming] $c$. Then $R_{c}\setminus d_{i}$ implies $c$, inducing $P(k+1)$ as shown previously.
				
				\item
				If $d_{i|I}\in F_{\phi|I}$ subsumes $c_{|I}$ then $P(k+1)$ is satisfied
				for $c$.
				
				\item 
				If $d_{i|I}$ does not subsume $c_{|I}$ and does not contain $\bar \alpha$ or $\bar \beta$ or $\bar \gamma$ then either (a) $d_{i|I}=(x)$
				or (b) $d_{i|I}=(ax)$ or (c) $d_{i|I}=(xy)$, where $x$ and
				$y$ $\notin\{\alpha \beta \gamma \}$ and are not set by $I$, and $a \in\{\alpha \beta \gamma \}$.
				
				\begin{enumerate}
				
					\item 
					If $d_{i|I}=(x)$ then $R_{c}\setminus d_{i}$ implies $(\bar{x}\alpha \beta \gamma L_{I})$ (recall that implying a certain clause $C$ means implying a clause which subsumes $C$). Since any resolution with $d_{i|I}=(x)$ as operand removes $\bar{x}$ from the other operand then no clause of $R_{c}\setminus d_{i}$ contains $\bar x$ (for $R_{c}\setminus d_{i} \subseteq R_{|I}$ which is the 3--limited closure of $F_{\phi}\cup F_{\phi|I}$). Then $R_{c}\setminus d_{i}$ implies $(\alpha \beta \gamma L_{I})$, inducing $P(k+1)$ as shown in Point (4). 
					
					\item 
					If $d_{i|I}=(ax)$ then $R_{c}\setminus d_{i}$ implies $(\bar{x}\alpha \beta \gamma L_{I})$. Replace $\bar{x}$ by $a$ in each possible
					clause of $R_{c}\setminus d_{i}$ (if the new clause is subsumed by some clause in $R_{|I}$, keep the subsuming clause instead. Anyway, the replacing clause is in $R_{|I}$). Name $R_{c,d_{i}}$ the resulting
					set ($R_{c,d_{i}}\subseteq R_{|I}$). Then $R_{c,d_{i}}$ implies $(\alpha \beta \gamma L_{I})$, inducing $P(k+1)$
					as above.
					
					\item 
					If $d_{i|I}=(xy)$ then $R_{c}\setminus d_{i}$ implies $(\bar{x}\alpha \beta \gamma L_{I})$
					and $(\bar{y}\alpha \beta \gamma L_{I})$. Replace $\bar{x}$
					by $y$ in each possible clause of $R_{c}\setminus d_{i}$ (as above, if the new clause is subsumed by some clause in $R_{|I}$, keep the subsuming clause instead). Name $R_{c,d_{i}}$
					the resulting set ($R_{c,d_{i}}\subseteq R_{|I}$). Then $R_{c,d_{i}}$ implies $({y}\alpha \beta \gamma L_{I})$. Since
					it implies also $(\bar{y}\alpha \beta \gamma L_{I})$ then it implies the resolvent
					$(\alpha \beta \gamma L_{I})$, inducing $P(k+1)$. \\
				
				\end{enumerate}
			
			\end{enumerate}
			
			By this recurrence, any clause $\in$ the 3--closure of $F_{\phi|I}$
			is subsumed by some clause $\in$ 3L--$F_{\phi|I}^{c}$ (the other
			way holds as well). Then the 3--limited closure of $F_{\phi|I}$ (computable
			in polynomial time) corresponds to the 3--closure of $F_{\phi|I}$.
		
		\end{proof}

	\subsection
	{$F_{\phi|I}$ is unsatisfiable iff its 3--closure contains the empty
	clause}
	
		Given $I$, a partial assignment of $k$ variables among $(x_{i})\ (k\le n)$.
		
		\begin{proposition} 
			Given $F$, a 3--CNF formula on $n$ variables
			$(x_{i})$$_{i\le1\le n}$. $F$ is closed under resolution implies
			$F_{|I}$ is closed under resolution. 
		\end{proposition}
		
		\begin{proof}
			If $c_{1}\ni x_{i}$ and $c_{2}\ni\bar{x}_{i}$ are in $F_{|I}$ (in
			particular, $x_{i}$ is unset by $I$), pick clauses $d_{1},d_{2}$
			in $F$ which restrict to $c_{1}$ and $c_{2}$, respectively. Then
			$x_{i}\in d_{1}$ and $\bar{x_{i}}\in d_{2}$, hence their resolvent
			$(d_{1}\smallsetminus\{x_{i}\})\cup(d_{2}\smallsetminus\{\bar{x_{i}}\})$
			is subsumed by some $d\in F$. If $d$ contains a literal made $true$
			under $I$, then so does $d_{1}$ or $d_{2}$, contradicting their
			choice. Thus, $d_{|I}$ is in $F_{|I}$, and it subsumes $(c_{1}\smallsetminus\{x_{i}\})\cup(c_{2}\smallsetminus\{\bar{x_{i}}\})$.\\
			 {\em Thanks to Emil Je\u{r}\'abek} (\url{http://cstheory.stackexchange.com/a/16835/6346}).
		\end{proof}
		
		\begin{proposition} 
			$F_{\phi|I}$ is unsatisfiable iff its 3--closure
			contains the empty clause. 
		\end{proposition}
		
		\begin{proof}
			As $F_{\phi}^{c}=F_{\phi}\cup F_{\phi}^{r}$ then $F_{\phi|I}^{c}=F_{\phi|I}\cup F_{\phi|I}^{r}$.
			Suppose the 3--closure of $F_{\phi|I}$ is unsatisfiable (the other
			implication is obvious). Then $F_{\phi|I}^{c}$ is unsatisfiable and
			it contains the empty clause (from the previous proposition and the
			Quine's theorem \cite{quine1952}: A formula closed under resolution
			is unsatisfiable iff it contains the empty clause). 
			\begin{enumerate}
				\item 
				As $F_{\phi}^{c}$ is equivalent to $F_{\phi}$ then $F_{\phi|I}^{c}$
				is equivalent to $F_{\phi|I}$. 
				
				\item 
				Two equivalent formulas have the same 3-closure. 
	
				\item 
				If the empty clause is in a formula then it is in its 3-closure (since
				$|(\varnothing)|=0$). 
			\end{enumerate}
			Hence $(\varnothing)$ is in the 3--closure of $F_{\phi|I}$. 
		\end{proof}

	\subsection
	{$\bar{\phi}$, a set of partial assignments subsuming all assigments
	$\notin\phi$, can be computed in polynomial time}
	
		Consider some total order among the $n$ variables, say the lexicographic
		one.
		
		\begin{definition} 
			Some additionnal usefull definitions: 
			\begin{enumerate}
				\item 
				Let $M_{k}$ be the set of all 2$^{k}$ partial assignments $(I_{k})$
				on the first $k$ values of the variables ($1\le k\le n$). 

				\item 
				Let $\phi_{k}=\{I_{k}\in M_{k}$/$I_{k}\in\phi\}$ 

				\item 
				Let $\bar{\phi}_{k}=\{I_{k}\in M_{k}$/$I_{k-1}\in\phi_{k-1}$ and
				$I_{k}\notin\phi_{k}\}$ ($I_{0}=\varnothing$ and $\phi_{0}$ is
				the empty set) 

				\item 
				Let $\bar{\phi}=\bigcup_{k}\bar{\phi}_{k}$ 

				\item 
				Let $m_{i,j}$ the restriction of $m_{i}\in\phi$ to its first $j$
				values and $\bar{m}_{i,j}$ the restriction of $m_{i}\in\phi$ to
				its first $j-1$ values ($j\ge1$) concatenated with the opposite
				of its $j^{th}$ value (as last value). 
			\end{enumerate}
		\end{definition}
		
		\begin{proposition} 
			About $\bar{\phi}_{k}$ 
			\begin{enumerate}
				\item 
				The extension to the rest of the $n$ variables of any partial assignment
				of $\bar{\phi}_{k}$ is not in $\phi$. 

				\item 
				An assignment $I_{n}$ of the $n$ variables does not belong to $\phi$
				iff $\exists k\le n$, $I_{k}\in\bar{\phi}_{k}$ where $I_{k}$ is
				the partial assigment issued from $I_{n}$ restricted to the first
				$k$ values.
 
				\item 
				The computation of $\bar{\phi}_{k}$ can be done in polynomial time. 
			\end{enumerate}
		\end{proposition}
		
		\begin{proof}
			$(1)$ Since any element of $\bar{\phi}_{k}$ is not in $\phi$, neither
			is any extension of it.\\
			 $(2)$ If $I_{n}\notin\phi$ then obviously $\exists k\le n$,
			$I_{k}\in\bar{\phi}_{k}$. If $\exists k\le n$, $I_{k}\in\bar{\phi}_{k}$
			where $I_{k}$ is the partial assigment issued from $I_{n}$ restricted
			to the first $k$ values then by $(1)$ any extension of $I_{k}\notin\phi$
			and $I_{n}\notin\phi$.\\
			 $(3)$ $|\phi_{k}|,|\bar{\phi}_{k}|\le|\phi|$ (and $|\bar{\phi}|\le n|\phi|$).
			The computation of $\phi_{k}$ can obviously be done in polynomial
			time. So can be the computation of $\bar{\phi}_{k}$: for each model
			$m_{i}\in\phi$, compute $\bar{m}_{i,k}$, put it in $\bar{\phi}_{k}$
			if it does not belong to $\phi_{k}$.
		\end{proof}
		
		\begin{proposition}
			About $\bar{\phi}$ 
			\begin{enumerate}
				\item 
				The extension to the rest of the $n$ variables of any partial assignment
				of $\bar{\phi}$ is not in $\phi$. 

				\item 
				$\bar{\phi}$ is a set of partial assignments subsuming all assigments
				of the $n$ variables which are not in $\phi$ ($|\bar{\phi}|\le n|\phi|$). 

				\item 
				$\bar{\phi}$ can be computed in polynomial time. 
			\end{enumerate}
		\end{proposition}
		
		\begin{proof} 
			Directly from the previous proposition and the definition
			of $\bar{\phi}$. 
		\end{proof}
		
		\begin{remark*}
			As we are interested in partial assignments which could be extended
			to an entire model for the 3--CNF $F$, we can only consider the $\bar{\phi}_{k}$
			sets for $k>3$ without changing anything further.
		\end{remark*}
		
		\begin{example}
			Take $n=5$ and 8 models $(m_{i})_{1\leq i\leq8}$ in $\phi$.\\
			 $\phi=\{00111,01011,10101,11100,11111,10011,01101,00100\}$ (as
			Example \ref{ex1})\\
			 The 3--closure of the candidate formula has been established:\\
			 $F_{\phi}=(x_{1}x_{2}x_{3})(\bar{x}_{1}\bar{x}_{2}x_{3})(x_{1}\bar{x}_{2}x_{5})(\bar{x}_{1}x_{2}x_{5})(x_{3}x_{4})(x_{3}x_{5})(\bar{x}_{4}x_{5})$\\

			Let build the sets $(\bar{\phi})_{k}$ for $4\le k\le n(=5)$:
			\begin{itemize}
			\item {$k=4$\\
			 $\phi_{4}=\{0011,0101,1010,1110,1111,1001,0110,0010\}\ \bar{m}_{1,4}=0010\in\phi_{4}$
			($=m_{8,4}$ so $\bar{m}_{8,4}=m_{1,4}\in\phi_{4}$)\\
			 $\bar{m}_{2,4}=0100\notin\phi_{4}$ ($\in\bar{\phi}_{4}$)\\
			 and so on until $\bar{\phi}_{4}=\{0100,1011,1000,0111\}$} 
			\item $k=5$\\
			 In the same way, $\bar{\phi}_{5}=\{00110,01010,10100,11101,11110,10010,01100,00101\}$ 
			\end{itemize}
			Hence $\bar{\phi}=\bar{\phi}_{4}\cup\bar{\phi}_{5}$
		\end{example}

	\subsection
	{An equivalent formulation of the (co)Inverse 3--SAT question: Is
	there a partial assignment $I\in\bar{\phi}$ such that the 3--closure
	of $F_{\phi|I}$ does not contain the empty clause ?}
	
		\begin{proposition} 
			The (co)Inverse 3--SAT question \textquotedbl{}Is
			there a model of $F_{\phi}^{3}\notin\phi$ ?\textquotedbl{} is equivalent
			to the question \textquotedbl{}Is there a partial assignment $I\in\bar{\phi}$
			such that the 3--closure of $F_{\phi|I}$ does not contain the empty
			clause ?\textquotedbl{} 
		\end{proposition}
		
		\begin{proof}
			If it exists a partial assignment $I\in\bar{\phi}$ such that the
			3--closure of $F_{\phi|I}$ does not contain the empty clause then
			: \\
			 1) All extensions of $I$ on the rest of the $n$ variables are
			not in $\phi$ (from Prop. 3.4).\\
			 2) $F_{\phi|I}$ is satisfiable (from Prop. 3.2).\\
			 Then $I$ extended (concatenated) with a model of $F_{\phi|I}$
			is a model of $F_{\phi}^{3}\notin\phi$.\\
			 If it exists $m$, a model of $F_{\phi}^{3}\notin\phi$ ($m$
			is also a model of $F_{\phi}$) then it exists a partial assignment
			$I_{m}\in\bar{\phi}$ which subsumes $m$ (since $\bar{\phi}$ is
			a set of partial assignments which subsume all assignment $\notin\phi$).
			Then $F_{\phi|I_{m}}$ is satisfiable (if not, no extension of $I_{m}$
			can satisfy neither $F_{\phi}$ nor $F_{\phi}^{3}$: contradiction)
			and its 3--closure does not contain the empty clause.
		\end{proof}

\section{The algorithm}

	\paragraph*
	{Input:\textmd{ $\phi$, a set of models over $n$ variables.} }
	
	\paragraph*
	{Step 1: \textmd{Compute $F_{\phi}$, the 3--closure of the candidate
	formula.} }
	
	\paragraph*
	{Step 2: \textmd{Compute $\bar{\phi}$, a set of partial assignments
	subsuming all assigments $\notin\phi$.} }
	
	\paragraph*
	{Step 3: \textmd{For each partial assignment $I\in\bar{\phi}$, compute
	the 3--closure of $F_{\phi|I}$ and check whether it contains the
	empty clause.} }
		
	\paragraph*
	{Output: \textmd{Yes or No, answering the question: Is there a partial
	assignment $I\in\bar{\phi}$ such that the 3--closure of $F_{\phi|I}$
	does not contain the empty clause ? }}
	
	\begin{proposition} 
		This algorithm lets solve the (co)Inverse 3--SAT
		problem. Each step can be computed in polynomial time. 
	\end{proposition}
	
	\begin{proof}
		This algorithm obviously finishes. It outputs the answer to the question:
		Is there a partial assignment $I\in\bar{\phi}$ such that the 3--closure
		of $F_{\phi|I}$ does not contain the empty clause ? which is equivalent
		to the classical (co)Inverse 3--SAT question. Its polynomial-time
		computation comes directly from the previous results of the article
		(since $|\bar{\phi}|\le n|\phi|$, there is no exponential increase
		in size).
	\end{proof}
	
	\begin{example}
		Take $n=5$ and 8 models $(m_{i})_{1\leq i\leq8}$ in $\phi$.\\
		 $\phi=\{00111,01011,10101,11100,11111,10011,01101,00100\}$ (as
		Example \ref{ex1} and 2)\\
		 $F_{\phi}$ and $\bar{\phi}$ have been found: \\
		 $F_{\phi}=(x_{1}x_{2}x_{3})(\bar{x}_{1}\bar{x}_{2}x_{3})(x_{1}\bar{x}_{2}x_{5})(\bar{x}_{1}x_{2}x_{5})(x_{3}x_{4})(x_{3}x_{5})(\bar{x}_{4}x_{5})$\\
		 $\bar{\phi}=\{0100,1011,1000,0111,00110,01010,10100,11101,11110,10010,01100,00101\}$\\
		 $F_{|0100}=(\varnothing)$ but $F_{|1011}=(x_{5})$ so the candidate
		formula has at least one model $m\notin\phi$ ($m=10111$).
	\end{example}

\section{Conclusion}

	The (co)Inverse 3--SAT problem can be solved in polynomial time.


\bibliographystyle{amsplain} %    Insert the bibliography data here.
\bibliographystyle{amsplain}
\bibliography{InvSAT_Labouze}

\end{document}